\documentclass[journal,draftcls,onecolumn,12pt,twoside]{IEEEtranTCOM}

\normalsize

\usepackage[noadjust]{cite}
\usepackage{bm}
\usepackage{amsfonts}
\usepackage{indentfirst}
\usepackage{algorithmic}
\usepackage{array}
\usepackage{mdwmath}
\usepackage{amsmath}
\usepackage{times}
\usepackage{mathptmx}
\DeclareMathAlphabet{\mathcal}{OMS}{cmsy}{m}{n}
\DeclareSymbolFont{largesymbols}{OMX}{cmex}{m}{n}
\usepackage{mdwtab}
\usepackage{multirow}
\usepackage{multicol}
\usepackage{amsbsy}
\usepackage[caption=false,font=footnotesize]{subfig}
\usepackage{graphicx}
\usepackage{epsfig,latexsym}
\usepackage{balance}
\usepackage{flushend}
\usepackage{xcolor,stfloats}
\usepackage{epstopdf}
\usepackage{setspace}
\usepackage{geometry}

\newtheorem{lemma}{Lemma}
\newtheorem{theorem}{Theorem}

\renewcommand{\baselinestretch}{1.4}
\IEEEoverridecommandlockouts

\begin{document}

\title{\LARGE Maritime Coverage Enhancement Using UAVs \\Coordinated with Hybrid Satellite-Terrestrial Networks}

\author{\normalsize Xiangling~Li,
Wei~Feng,~\IEEEmembership{Senior Member,~IEEE,}
Yunfei~Chen,~\IEEEmembership{Senior Member,~IEEE,}\\
Cheng-Xiang~Wang,~\IEEEmembership{Fellow,~IEEE,}
Ning~Ge,~\IEEEmembership{Member,~IEEE}
        \thanks{X.~Li, W.~Feng (corresponding author), and N. Ge are with the Beijing National Research Center for Information Science and Technology, Tsinghua University, Beijing 100084, China. W. Feng is also with the Peng Cheng Laboratory, Shenzhen 518000, China. Y. Chen is with the School of Engineering, University of Warwick, Coventry CV4 7AL, U.K. C.-X. Wang is with the National Mobile Communications Research Laboratory, School of Information Science and Engineering, Southeast University, Nanjing 210096, China, and also with Purple Mountain Laboratories, Nanjing 211111, China. (e-mail: lingjlu@yeah.net, {fengwei}@tsinghua.edu.cn, {Yunfei.Chen}@warwick.ac.uk, {chxwang}@seu.edu.cn, {gening}@tsinghua.edu.cn).

        Part of this work has been accepted by IEEE WOCC'2019 \cite{IEEEhowto:kopka52}.}
        }


\maketitle
\begin{spacing}{1.58}
\begin{abstract}
Due to its agile maneuverability, unmanned aerial vehicles (UAVs) have shown great promise for on-demand communications. In practice, UAV-aided aerial base stations are not separate. Instead, they rely on existing satellites/terrestrial systems for spectrum sharing and efficient backhaul.
In this case, how to coordinate satellites, UAVs and terrestrial systems is still an open issue. In this paper, we deploy UAVs for coverage enhancement of a hybrid satellite-terrestrial maritime communication network. Under the typical composite channel model including both large-scale and small-scale fading, the UAV trajectory and in-flight transmit power are jointly optimized, subject to constraints on UAV kinematics, tolerable interference, backhaul, and the total energy of UAV for communications. Different from existing studies, only the location-dependent large-scale channel state information (CSI) is assumed available, because it is difficult to obtain the small-scale CSI before takeoff in practice, and the ship positions can be obtained via the dedicated maritime Automatic Identification System.
The optimization problem is non-convex. We solve it by problem decomposition, successive convex optimization and bisection searching tools. Simulation results demonstrate that the UAV fits well with existing satellite and terrestrial systems, using the proposed optimization framework.
\end{abstract}

\end{spacing}
\IEEEpeerreviewmaketitle

%
%
%
%
\begin{spacing}{1.9}
\section{Introduction}
Currently, various activities on the ocean increase, leading to the growing demands for wireless communications~\cite{IEEEhowto:kopka52, mari1, mari2}. To satisfy the increasing requirements, hybrid satellite-terrestrial networks emerge, in which satellites and terrestrial systems are integrated for a better maritime coverage \cite{IEEEhowto:kopka24,IEEEhowto:kopka40,IEEEhowto:kopka41}. Basically, the satellites, deployed in the Geostationary Earth Orbit or Low Earth Orbits, can provide a wide-area coverage \cite{IEEEhowto:kopka7}. However, their transmission rate is usually limited due to long transmission distance and restricted onboard payloads. High-throughput satellites have thus been attracting great attentions \cite{IEEEhowto:kopka8}. Yet, it is still quite challenging to realize the global broadband coverage using the state-of-the-art satellite technologies at a practically affordable cost. As an alternative, the terrestrial base stations (TBSs) can be deployed along the coastline to offer high-rate communication services. However, their coverage range is usually limited.

Different from satellites and TBSs, unmanned aerial vehicles (UAVs) have shown considerable promise for agile communications \cite{IEEEhowto:kopka53,IEEEhowto:kopka4}. UAVs can enable aerial base stations with largely increased line of sight (LOS) transmission range. Moreover, UAVs can adaptively change their spatial locations according to the communication demands.
While most existing studies on UAVs focused on the terrestrial scenario, we explore the potential gain of UAVs for maritime coverage enhancement in this paper. Particularly, we focus on the coordination issue between introduced UAVs and existing maritime satellites and terrestrial systems.

Related studies can be categorized into three types according to the considered system model, which are summarized as follows.

\subsubsection{UAVs only}
Most previous works focused on the UAV-only system model, while ignoring satellites and TBSs. For rotary-wing UAVs, the optimal placement of UAVs has been widely investigated, leading to many insightful observations~\cite{IEEEhowto:kopka6,IEEEhowto:kopka31,IEEEhowto:kopka2,IEEEhowto:kopka9,IEEEhowto:kopka22,IEEEhowto:kopka45,IEEEhowto:kopka37}.
For fixed-wing UAVs, the trajectory design is an important issue, which is closely related to the UAV's kinematic parameters.
Considering the UAV's maximum velocity and/or maximum acceleration, the trajectory of UAV was optimized for maximum throughput or minimum UAV periodic flight duration, or
optimal energy efficiency \cite{IEEEhowto:kopka33,IEEEhowto:kopka3,IEEEhowto:kopka29,IEEEhowto:kopka14,IEEEhowto:kopka35,IEEEhowto:kopka27}.
These works~\cite{IEEEhowto:kopka6,IEEEhowto:kopka31,IEEEhowto:kopka2,IEEEhowto:kopka9,IEEEhowto:kopka22,IEEEhowto:kopka45,IEEEhowto:kopka37,IEEEhowto:kopka33,IEEEhowto:kopka3,IEEEhowto:kopka29,IEEEhowto:kopka14,IEEEhowto:kopka35,IEEEhowto:kopka27} mainly considered static users. For mobile users, the ergodic achievable rate was maximized by dynamically adjusting the UAV heading~\cite{IEEEhowto:kopka28,IEEEhowto:kopka17,IEEEhowto:kopka36}. Intuitively in the maritime scenario, the UAV's trajectory should adaptively cater to the mobility of ships, providing an accompanying broadband coverage, which however remains elusive.

\subsubsection{Coexistence of UAVs and TBSs}
In addition to UAV-only models, the coexistence of UAVs and TBSs was investigated in \cite{IEEEhowto:kopka30,IEEEhowto:kopka26,IEEEhowto:kopka38,IEEEhowto:kopka34, IEEEhowto:kopka46}.
For rotary-wing UAVs, the TBS can be used as a hub to connect UAVs to the network \cite{IEEEhowto:kopka30}. In this case, the access link and backhaul link should be jointly optimized to maximize the sum rate. In \cite{IEEEhowto:kopka26}, the UAV-based multi-hop backhaul network was formulated to adapt to the dynamics of the network. Outage probability is also an important issue for the coexistence of UAVs and TBSs~\cite{IEEEhowto:kopka38,IEEEhowto:kopka34,IEEEhowto:kopka46}. In \cite{IEEEhowto:kopka34}, the outage probability was minimized. In \cite{IEEEhowto:kopka46}, the throughput was maximized subject to the maximum outage probability constraint. For the maritime scenario, the TBS is the primary choice for UAV backhaul, due to their high-speed transmission rate.

\subsubsection{Coexistence of UAVs and Satellites}
More recently, the integration of UAVs and satellites has been investigated in \cite{IEEEhowto:kopka10, IEEEhowto:kopka32, IEEEhowto:kopka12,IEEEhowto:kopka19,IEEEhowto:kopka25,IEEEhowto:kopka13}. Particularly,
the authors of \cite{IEEEhowto:kopka10} investigated the integration of satellite and UAV communications for heterogeneous
flying vehicles. In addition, the long transmission delay is quite challenging for satellites. Thus in \cite{IEEEhowto:kopka12}, the impact of UAV altitude on the average delay was analyzed to coordinate UAVs and satellites. A multi-UAV assisted network was formulated in \cite{IEEEhowto:kopka19}, where the coverage probability and the ergodic achievable rate were analyzed for post-disaster areas.
The airborne mobile wireless networks were considered in \cite{IEEEhowto:kopka25}, where an efficient power allocation scheme was proposed to support the diverse real-time services.

Despite of the aforementioned interesting works, there are still open problems to integrate UAVs into hybrid satellite-terrestrial maritime communication networks.
Firstly, to solve the spectrum scarcity problem, it is valuable to explore the potential of spectrum sharing among satellites, UAVs and terrestrial networks. Till now, spectrum sharing between satellites and terrestrial networks has been studied \cite{IEEEhowto:kopka42,IEEEhowto:kopka43,IEEEhowto:kopka44}.
For more complicated spectrum sharing among satellites, UAVs and terrestrial networks, it is a crucial issue to obtain the channel state information (CSI) for interference mitigation.
Both the large delay of satellite transmission and the mobility of UAVs and ships render this issue challenging.
Secondly, before takeoff, a whole trajectory of UAVs needs to be planned for coverage enhancement according to the mobility of targeted ship.
However, the limited capacity of wireless backhaul affects the real-time transmission, and the communication energy of UAVs provided by the batteries is also limited. These constraints should be considered in the optimization of UAV trajectory.
Besides, different from most previous works which use the free space path loss model to simplify analysis, it is more practical to consider both large-scale and small-scale fading \cite{feng13,feng17}. However, it is difficult to acquire the random small-scale fading before takeoff~\cite{IEEEhowto:kopka39}.

Motivated by the above observations, we investigate a hybrid satellite-UAV-terrestrial maritime communication network where UAVs are integrated for coverage enhancement. Considering
the severe environment on the ocean, we consider the fixed-wing UAV, which has longer duration of flight and stronger anti-wind capability than the rotary-wing UAV. In our model, the UAV shares spectrum with satellites, and utilizes TBSs for wireless backhaul. A typical composite channel model including both large-scale and small-scale fading is used.
We obtain the ship positions from the dedicated maritime Automatic Identification System. Accordingly, quite different from the terrestrial scenario, we assume that the large-scale CSI
is available before UAV takes off. Because the large-scale CSI is location dependent, and we can obtain it with historical or pre-measured data. We optimize
the whole trajectory and transmit power during the fight, subject to the UAV's kinematical constraints, the backhaul constraints, tolerable interference constraints and the communication energy. The optimization problem is non-convex. We decompose the problem and solve it by successive convex optimization and bisection searching tools. Simulation results demonstrate that the UAV fits well with existing satellite and terrestrial systems. Besides, a significant performance gain can be achieved via joint optimization of the UAV's trajectory and transmit power by using only the large-scale CSI.

The rest of this paper is organized as follows. In Section II, the system model is introduced. The problem for the UAV-aided coverage enhancement is formulated and solved in Section III. In Section IV, simulation results are presented. Section V concludes the paper.

Throughout this paper, the vectors and scalars are denoted by boldface letters, and normal letters, respectively. ${|\cdot|}$ indicates the absolute value of a scalar. Transpose operator is indicated with $[\cdot]^T$. $\ell_p$-norm means ${\left\| {\bm x} \right\|_p} = {\left( {{{\sum\nolimits_{i = 1}^n {\left| {{x_i}} \right|} }^p}} \right)^{{1 \mathord{\left/ {\vphantom {1 p}} \right.\kern-\nulldelimiterspace} p}}}$. ${\mathcal{CN}(0,\sigma^2)}$ represents the complex Gaussian distribution with zero mean and $\sigma^2$ variance. $\dot{{\bm x}}_t$ and $\ddot{{\bm x}}_t$ denote the first-order and second-order derivatives of ${\bm x}_t$ with respect to $t$. ${\bf{E}\{\cdot\}}$ denotes the expectation operator.

\section{System Model}
We consider a practical hybrid maritime network consisting of mobile users (ships), UAVs, TBSs and satellites, as shown in Fig. \ref{fig_sys_model}. The TBSs are deployed along the coastline to provide communication services for users in the area of coastal waters. The broadband coverage area of TBSs is usually limited due to large non-line-of-sight pathloss.
Out of the coverage area of TBSs, the maritime satellites provide communication links. For the ships equipped with expensive high-gain antennas, the broadband service can be guaranteed. Whereas for the low-end ships without high-gain antennas, it is still difficult to enjoy a broadband service even within the coverage area of satellites. To fill up the blind holes, we utilize UAVs to provide broadband services in an on-demand manner. When a user requests high-rate communications, a UAV will be sent out to provide that service. Otherwise, the UAV waits near a TBS.

In this paper, the spectrum is shared between UAVs and satellites. Thus, there may be the interference between the UAV-to-user link and the satellite-to-user link. Because the antenna gain of the users served by UAVs is lower than that of the users served by satellites, the interference on the users served by UAVs from satellites can be ignored. To mitigate the interference on the users served by satellites, we can adjust the trajectory and the transmit power of UAVs.

\begin{figure*}[!t]
  \centering
  \includegraphics[width=6.5in]{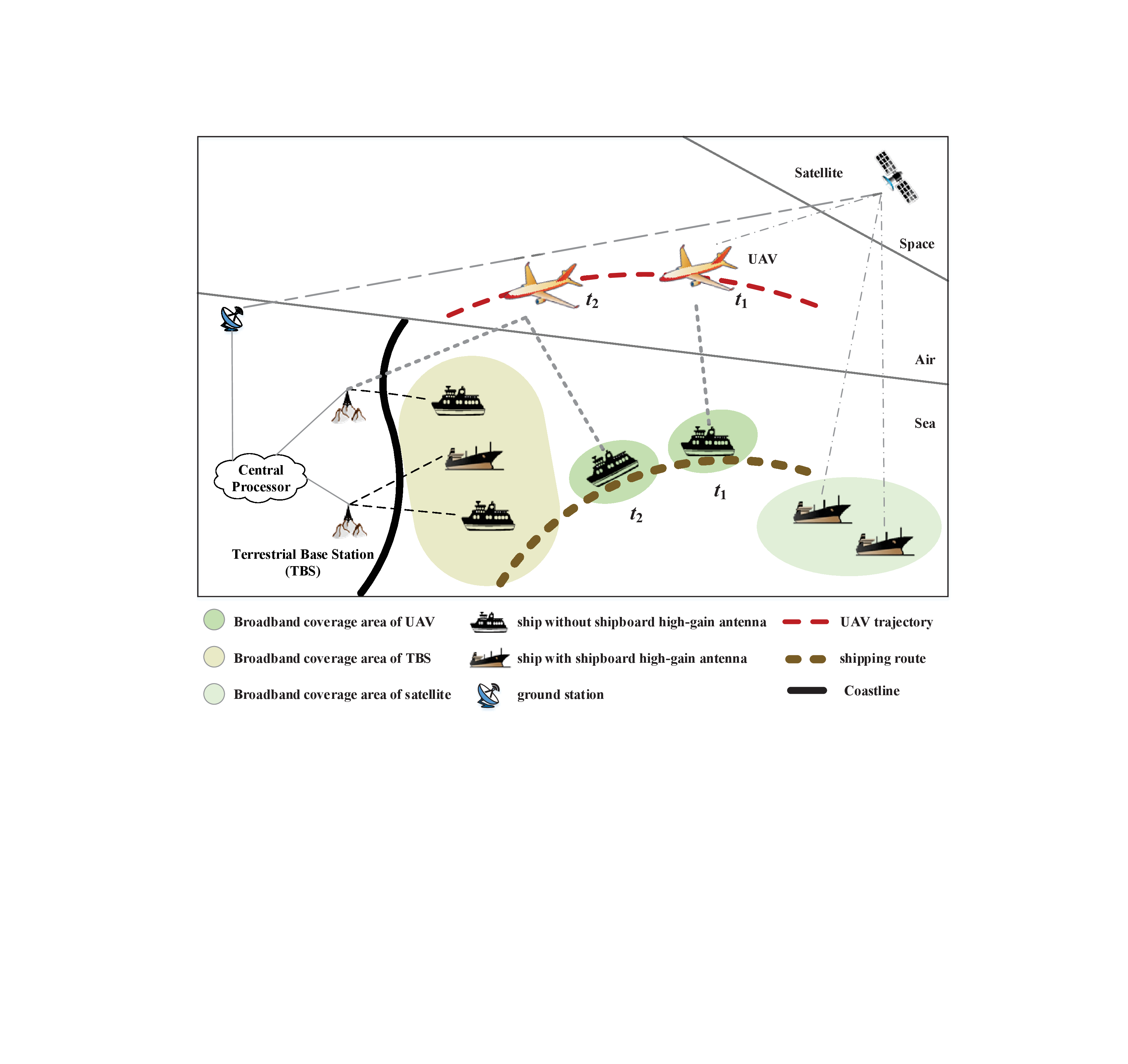}%
  \caption{Illustration of a hybrid satellite-UAV-terrestrial maritime communication network, where satellites, UAVs and TBSs provide broadband services in a coordinated manner.}\label{fig_sys_model}
\end{figure*}

To serve the mobile users on the ocean, UAVs need the wireless backhaul. Both TBSs and satellites can be used. As shown in \cite{IEEEhowto:kopka10}, when UAVs are close to the mainland, the air-to-ground backhaul is able to provide enough capacity. In this case, the TBSs nearest to UAVs could be utilized to connect UAVs to the central processor. Otherwise, satellites are used instead.

We assume that autonomous UAVs are employed as the aerial base stations. The UAVs have the abilities of dynamic mission plan, inter-cell handover, resource allocation, etc.
Let $T_0$ be the travel time during which a user is served by a UAV. Considering the user mobility, our aim is to maintain certain achievable rate to avoid severe performance degradation during the travel time. Before the UAV serves the user, the trajectory and the transmit power of the UAV are optimized to maximize the minimum ergodic rate during the travel time ${T_0}$.

The ergodic achievable rate $R_{i,j,t}$ between the $i$-th transmitter and the $j$-th receiver at time $t$ can be denoted as
\begin{equation}\label{eqn_1}
R_{i,j,t}={\bf{E}}\left\{ \log_2 \left[ {1 + \frac{{{P_{i,t}}{G_iG_j}|h_{i,j,t}|^2}}{{\sigma ^2}}} \right]\right\}
\end{equation}
where ${h_{i,j,t}} $ denotes the channel between the $i$-th transmitter and the $j$-th receiver at time $t$, and $P_{i,t}$ denotes the transmit power, and ${\sigma ^2}$ denotes the white Gaussian noise power, and $G_{i}$ denotes the gain of the transmitting antenna, and $G_{j}$ denotes the gain of the receiving antenna. The expectation is taken over the small-scale fading.

We assume that both UAVs and users are equipped with a single antenna, and UAVs are high enough to enable LOS transmission. A typical composite channel containing both large-scale and small-scale fading is employed. The channel between the $i$-th transmitter and the $j$-th receiver at time $t$ can be represented as
\begin{equation}
{h_{i,j,t}} = L_{i,j,t}^{{{{\rm{ - }}1} \mathord{\left/ {\vphantom {{{\rm{ - }}1} 2}} \right.
 \kern-\nulldelimiterspace} 2}}{\tilde h_{i,j,t}}
\end{equation}
where $L_{i,j,t}$ denotes the path loss, and ${\tilde h}_{i,j,t}$ denotes Rician fading.
Let ${d_{i,j,t}}$ denote the distance between the $i$-th transmitter and the $j$-th receiver at time $t$. We assume the earth surface to be smooth and flat\footnote{If the distances are shorter than a few tens of kilometers, it is often permissible to neglect earth curvature and assume the earth surface to be smooth and flat \cite{IEEEhowto:kopka5}.}. Then, the path loss model can be expressed as
\begin{equation}\label{eqn_10}
{L_{i,j,t}}\left( {{\rm{dB}}} \right) ={A_0} + 10\varsigma \log 10\left( {\frac{{{d_{i,j,t}}}}{{{d_0}}}} \right) + {X_{i,j,t}}
\end{equation}
where $d_0$ denotes the reference distance, and $A_0$ denotes the path loss at $d_0$, and $\varsigma$ denotes the path-loss exponent, and $X_{i,j,t}$ indicates zero-mean Gaussian random variable with standard deviation $\sigma_X$ \cite{IEEEhowto:kopka20, IEEEhowto:kopka50, IEEEhowto:kopka51}.
Rician fading can be represented as
\begin{equation}
{{\tilde h}_{i,j,t}=\sqrt {\frac{K}{{1 + K}}}  + \sqrt {\frac{1}{{1 + K}}} {g_{i,j,t}}}
\end{equation}
where ${g_{i,j,t} \in \mathcal{CN}(0,~1)}$, and $K$ indicates the Rician factor that corresponds to the ratio between the LOS power and the multipath power \cite{IEEEhowto:kopka23, IEEEhowto:kopka47, IEEEhowto:kopka48, IEEEhowto:kopka49}. The path loss and Rician fading correspond to the large-scale and small-scale fading, respectively. The path loss is location dependent. On the ocean, the ships normally travel along a fixed shipping route, and hence the positions of ships can be obtained with historical or pre-measured data. We assume that the Rician factor $K$ is available to TBSs.

\section{UAV-Aided Coverage Enhancement}
In this section, we formulate the optimization problem of the UAV trajectory and in-flight transmit power and provide an iterative algorithm to solve the optimization problem.
\subsection{Problem Formulation}
The sets of TBSs and UAVs are denoted by $\Gamma_s$ and $\Gamma_a$, respectively. The sets of the users served by UAVs and satellites are denoted by $\Omega_a$ and $\Omega_o$, respectively. Without loss of generality, we assume that the UAV is connected to a fixed TBS during the travel time ${T_0}$. The association of the UAV to TBSs is not considered in this paper. We consider a three-dimensional Cartesian coordinate system, in which the TBS is located at ${(0,\ 0,\ z_{s,t})}$, where $s \in {\Gamma_s}$. The positions of the UAV and its user at time $t$ are respectively denoted as ${{\bf{c}}_{a,t}=[x_{a,t},~y_{a,t},~z_{a,t}]^T}$ and ${{\bf{c}}_{i,t}=[x_{i,t},~y_{i,t},~z_{i,t}]^T}$, where $a \in {\Gamma_a}$ and $i \in {\Omega _a}$. We discretize the travel time ${T_0}$ into ${T}$ time slots with a step size ${\Delta t}$. We adjust the trajectory and the transmit power of UAV per time slot. The step size ${\Delta t}$ can be set according to the variation of user's positions.

Let ${\Omega '_{o,t}}$ be the set of users served by satellites but sharing the same frequency with the user served by the UAV at time $t$, and ${|\Omega '_{o,t}|=M_t}$. To avoid the interference, an interference temperature limitation $I_0$ is applied in constraints and we have
\begin{equation}
\label{eqn_22_g} {\bf{E}}\left[ {{P_{a,t}}{G_aG_j}|{h_{a,j,t}}{|^2}} \right] \le {I_0}, {j\in \Omega '_o}.
\end{equation}
The expectation is taken over the small-scale fading.

We consider the air-to-ground backhaul. Due to the wireless backhaul, the ergodic achievable rate of the UAV-to-user link cannot exceed to that of the TBS-to-UAV link. Thus, we have
\begin{equation}
\label{eqn_22_f} {R_{a,i,t}} \leq {R_{s,a,t}}.
\end{equation}
The transmission distances between UAVs and satellites are quite large, so the transmission distances can be assumed to be a constant during a short travel time. Thus, when satellites are used for the backhaul, a constant can be used as the upper bound.

The definition of the velocity and the acceleration of the fixed-wing UAV can be expressed as
\begin{flalign}
\label{eqn_22_i} &{{\bf{v}}_{a,t}} = {{\bf{\dot c}}_{a,t}},\\
\label{eqn_22_j} &{{\bf{a}}_{a,t}} = {{\bf{\ddot c}}_{a,t}}.
\end{flalign}
Because of the existing boundary conditions on the velocity and the acceleration, we have
\begin{flalign}
\label{eqn_22_b}&{\left\| {{\bf{v}}_{a,t}} \right\|_2^2}\geq {{v}}_{\min }^2, \\
\label{eqn_22_k}&{\left\| {{\bf{v}}_{a,t}} \right\|_2^2} \le {{v}}_{\max }^2, \\
\label{eqn_22_c}&{\left\| {{\bf{a}}_{a,t}} \right\|_2^2} \le {{a}}_{\max }^2
\end{flalign}
where ${v_{\min }}$ denotes the minimum velocity, and ${v_{\max }}$ denotes the the maximum velocity, and ${a_{\max }}$ denotes the maximum acceleration. Besides, considering the boundary on the height of the UAV, we have
\begin{equation}
\label{eqn_22_d} {z_{\min }} \le {z_{a,t}} \le {z_{\max }}.
\end{equation}
The lower bound in (\ref{eqn_22_d}) is used to guarantee the LOS link. The upper bound in (\ref{eqn_22_d}) is set to indicate the maximum height that the UAV can reach according to the air traffic control.

We focus on the dynamic coverage performance of the user during $T$ time slots. As the energy consumption for communications is limited, we have
\begin{equation}
\label{eqn_22_h} \sum\nolimits_{t=1}^{T}{{P_{a,t}}\Delta t} \leq {E_0}
\end{equation}
where ${E_0}$ denotes the allowable energy consumption during $T_0$. Considering the maximum transmit power $P_{\max}$, we have
\begin{equation}
\label{eqn_22_e} 0 \leq {P_{a,t}} \leq P_{\max}.
\end{equation}
The working time of the UAV is mainly determined by the fuel for flying and the battery for the communication. We assume that the fuel of the fixed-wing UAV is large enough for the trip during the travel time ${T_0}$. If the residual energy is not enough to provide services after $T_0$, multi-UAV scheduling can be employed.

According to the above analysis, the optimization problem can be formulated as
\begin{flalign}
\label{eqn_22_a} \mathop {\max }\limits_{{P_{a,t}},{{\bf{c}}_{a,t}},{{\bf{v}}_{a,t}}, {{\bf{a}}_{a,t}}}~~\mathop {\min }\limits_t~~&{R_{a,i,t}}\\
\nonumber \rm{subject\ to}~~& (\ref{eqn_22_g}), (\ref{eqn_22_f}), (\ref{eqn_22_i}), (\ref{eqn_22_j}), (\ref{eqn_22_b}), (\ref{eqn_22_k}), (\ref{eqn_22_c}), (\ref{eqn_22_d}), (\ref{eqn_22_h}), (\ref{eqn_22_e})
\end{flalign}
where the minimum ergodic achievable rate during $T$ time slots is maximized, by optimizing the UAV's transmit power, three-dimensional coordinates, velocities and accelerations during $T$ time slots.

\subsection{An Iterative Solution}
The optimization problem in (\ref{eqn_22_a}) is difficult because the expectation is taken over the Rician fading in (\ref{eqn_1}), (\ref{eqn_22_g}) and (\ref{eqn_22_f}). To solve this problem, the relationship between ergodic achievable rate ${R_{a,i,t}}$ and ${a_{a,i,t}}$ is analyzed and the result is demonstrated in the following theorem, where
\begin{equation}
{a_{a,i,t}} = {P_{a,t}}{G_a}{G_i}L_{a,i,t}^{ - 1}{\sigma ^{ - 2}}.
\end{equation}

\begin{theorem}
The ergodic achievable rate ${R_{a,i,t}}$ is strictly concave and monotonically increasing with respect to the average SNR ${a_{a,i,t}}$.
\end{theorem}
\begin{proof}
See Appendix A.
\end{proof}

According to the monotonicity of the objective function, we equivalently simplify (\ref{eqn_22_a}) as
\begin{equation}\label{eqn_4}
\mathop {\max }\limits_{{P_{a,t}},{{\bf{c}}_{a,t}},{{\bf{v}}_{a,t}},{{\bf{a}}_{a,t}}} \;\mathop {\min }\limits_t \;\;\;\frac{{{P_{a,t}}{G_aG_i}{L_{a,i,t}^{-1}}}}{{\sigma ^2}}.
\end{equation}
Similarly, we equivalently simplify (\ref{eqn_22_f}) as
\begin{equation}\label{eqn_8}
\frac{{{P_{a,t}}{G_aG_i}{L_{a,i,t}^{-1}}}}{{\sigma ^2}} \le \frac{{{P_{s,t}}{G_sG_a}{L_{s,a,t}^{-1}}}}{{\sigma ^2}}
\end{equation} where ${P_{s,t}}$ denotes the transmit power of the TBS.
To deal with the derivatives in (\ref{eqn_22_i}) and (\ref{eqn_22_j}), by using the first-order and second-order Taylor approximations, the constraints in (\ref{eqn_22_i}) and (\ref{eqn_22_j}) can be expressed as
\begin{flalign}
&{{\bf{v}}_{a,t+1}} \approx {{\bf{v}}_{a,t}} + {{\bf{a}}_{a,t}}\Delta t,\\
& {{\bf{c}}_{a,t+1}} \approx {{\bf{c}}_{a,t}} + {{\bf{v}}_{a,t}}\Delta t + \frac{1}{2}{{\bf{a}}_{a,t}}\Delta {t^2}.
\end{flalign}
Let
\begin{flalign}
&\Delta {{\bf{v}}_t} = {{\bf{v}}_{a,t+1}} - ({{\bf{v}}_{a,t}} + {{\bf{a}}_{a,t}}\Delta t),\\
&\Delta {{\bf{c}}_t} = {{\bf{c}}_{a,t+1}} - \left( {{{\bf{c}}_{a,t}} + {{\bf{v}}_{a,t}}\Delta t + \frac{1}{2}{{\bf{a}}_{a,t}}\Delta {t^2}} \right).
\end{flalign}
We also let ${\Delta v_{w,t}}$ and ${\Delta c_{w,t}}$ respectively denote the $w$-th element in $\Delta {{\bf{v}}_t}$ and $\Delta {{\bf{c}}_t}$, where $w \in \{ 1,2,3\}$. We have
\begin{flalign}
&\label{eqn_19} \left| {\Delta {v_{w,t}}} \right| \le {\Delta {v_{0}}},\\
&\label{eqn_20} \left| {\Delta {c_{w,t}}} \right| \le {\Delta {c_{0}}}
\end{flalign}
where the thresholds $\Delta {v_{0}}$ and ${\Delta {c_{0}}}$ are set to be the small values.
According to ${g_{a,j,t} \in \mathcal{CN}(0,1)}$, we have
\begin{equation}
{\bf{E}}\left[ {{P_{a,t}}{G_aG_j}|{h_{a,j,t}}{|^2}} \right] = {P_{a,t}}{G_aG_j}{L_{a,j,t}^{-1}}.
\end{equation} Then, the constraint in (\ref{eqn_22_g}) can be rewritten as
\begin{equation}\label{eqn_23_c}
{P_{a,t}}{G_aG_j}{L_{a,j,t}^{-1}}\leq I_0.
\end{equation}
To solve the max-min problem, let
\begin{equation}
{\rm{Q}} = \mathop {\min }\limits_t~ {P_{a,t}}{G_a}{G_i}L_{a,i,t}^{ - 1}{\sigma ^{-2}}.
\end{equation}
Based on the above analysis, the problem in (\ref{eqn_22_a}) can be approximated as
\begin{subequations}\label{eqn_23}
\begin{flalign}
\mathop {\max }\limits_{{P_{a,t}},{{\bf{c}}_{a,t}},{{\bf{v}}_{a,t}},{{\bf{a}}_{a,t}},Q}\ \ \ & Q \\
\rm{subject\ to}\ \ \ & (\ref{eqn_22_b}), (\ref{eqn_22_k}), (\ref{eqn_22_c}), (\ref{eqn_22_d}), (\ref{eqn_22_h}), (\ref{eqn_22_e}),\nonumber\\
& (\ref{eqn_8}), (\ref{eqn_19}), (\ref{eqn_20}), (\ref{eqn_23_c}), \nonumber\\
\label{eqn_23_b} & Q \le \frac{{{P_{a,t}}{G_aG_i}{L_{a,i,t}^{-1}}}}{{\sigma ^2}},
\end{flalign}
\end{subequations}
In the constraints (\ref{eqn_8}), (\ref{eqn_23_c}) and (\ref{eqn_23_b}), the variables ${P_{a,t}}$ and ${{\bf{c}}_{a,t}}$ are in the numerator and denominator of the fractions, respectively. To make the analysis easy, based on the monotonicity of power functions, the constraints in (\ref{eqn_8}), (\ref{eqn_23_c}) and (\ref{eqn_23_b}) are rewritten as
\begin{flalign}
\label{eqn_21_a} & ({B_{s,t}}{P_{s,t}})^{\frac{2}{\varsigma}}{\left\| {{{\bf{c}}_{a,t}} - {{\bf{c}}_{i,t}}} \right\|_2^2}-({B_{i,t}}{P_{a,t}})^{\frac{2}{\varsigma}}{\left\| {{{\bf{c}}_{a,t}} - {{\bf{c}}_{s,t}}} \right\|_2^2}\ge 0, \\
\label{eqn_21_c} & {I_0}^{\frac{2}{\varsigma}}{\left\| {{{\bf{c}}_{a,t}} - {{\bf{c}}_{j,t}}} \right\|_2^2} \ge({B_{j,t}}{P_{a,t}})^{\frac{2}{\varsigma}},\\
\label{eqn_21_b} & Q^{\frac{2}{\varsigma}}{\left\| {{{\bf{c}}_{a,t}} - {{\bf{c}}_{i,t}}} \right\|_2^2}\le ({B_{i,t}}{P_{a,t}})^{\frac{2}{\varsigma}}
\end{flalign}
with
\begin{flalign}
&{B_{i,t}} = {G_a}{G_i}d_0^\varsigma{\sigma ^{ - 2}}{10^{ - \frac{{{A_0} + {X_{a,i,t}}}}{{10}}}},\\
&{B_{s,t}} = {G_s}{G_a}d_0^\varsigma{\sigma ^{ - 2}}{10^{ - \frac{{{A_0} + {X_{s,a,t}}}}{{10}}}},\\
&{B_{j,t}} = {G_a}{G_j}d_0^\varsigma{\sigma ^{ - 2}}{10^{ - \frac{{{A_0} + {X_{a,j,t}}}}{{10}}}}.
\end{flalign}
One can see that ${\left\| {{\bf{v}}_{a,t}}\right\|_2^2}$, ${\left\| {{\bf{a}}_{a,t}}\right\|_2^2}$, ${\left\| {{{\bf{c}}_{a,t}} - {{\bf{c}}_{i,t}}} \right\|_2^2}$ and ${\left\| {{{\bf{c}}_{a,t}} - {{\bf{c}}_{j,t}}} \right\|_2^2}$ are convex functions. The constraints in (\ref{eqn_22_k}), (\ref{eqn_22_c}) and (\ref{eqn_21_b}) indicate the convex sets with respect to ${{{\bf{v}}_{a,t}}}$, ${{{\bf{a}}_{a,t}}}$ and ${{{\bf{c}}_{a,t}}}$. The constraints in (\ref{eqn_22_b}) and (\ref{eqn_21_c}) indicate the concave sets with respect to ${{{\bf{v}}_{a,t}}}$ and ${{{\bf{c}}_{a,t}}}$.

Define the function
\begin{flalign}
f_1\left( {{\bf{c}}_{a,t}} \right) =({B_{s,t}}{P_{s,t}})^{{2 \mathord{\left/ {\vphantom {2 \varsigma }} \right.
 \kern-\nulldelimiterspace} \varsigma }}{\left\| {{{\bf{c}}_{a,t}} - {{\bf{c}}_{i,t}}} \right\|_2^2}-({B_{i,t}}{P_{a,t}})^{{2 \mathord{\left/ {\vphantom {2 \varsigma }} \right.
 \kern-\nulldelimiterspace} \varsigma }}{\left\| {{{\bf{c}}_{a,t}} - {{\bf{c}}_{s,t}}} \right\|_2^2}.
\end{flalign}
To determine the convexity of $(\ref{eqn_21_a})$, we verify the relationship between $f_1\left( {{\bf{c}}_{a,t}}\right)$ and ${{\bf{c}}_{a,t}}$ by the second-order derivatives.
We have the following theorem.
\begin{theorem}
If ${{{B_{s,t}}{P_{s,t}}\leq {B_{i,t}}{P_{a,t}}}}$, $f_1\left( {{\bf{c}}_{a,t}}\right)$ is a concave function, else if ${{{B_{s,t}}{P_{s,t}}>{B_{i,t}}{P_{a,t}}}}$, $f_1\left({{\bf{c}}_{a,t}}\right)$ is a convex function.
\end{theorem}
\begin{proof}
The second-order partial derivative of $f_1\left({{\bf{c}}_{a,t}}\right)$ with respect to ${{\bf{c}}_{a,t}}$ is
\begin{equation}
\ddot{f}_1\left({{\bf{c}}_{a,t}}\right)=2({B_{s,t}}{P_{s,t}})^{{2 \mathord{\left/ {\vphantom {2 \varsigma }} \right.
 \kern-\nulldelimiterspace} \varsigma }}- 2({B_{i,t}}{P_{a,t}})^{{2 \mathord{\left/ {\vphantom {2 \varsigma }} \right.
 \kern-\nulldelimiterspace} \varsigma }}.
\end{equation}
For any given ${B_{i,t}}$, ${B_{s,t}}$, ${P_{a,t}}$ and ${P _{s,t}}$, if ${{{B_{s,t}}{P_{s,t}}\leq {B_{i,t}}{P_{a,t}}}}$, $f_1\left({{\bf{c}}_{a,t}}\right)$ is a concave function, then we have a convex constraint in $(\ref{eqn_21_a})$. If ${{{B_{s,t}}{P_{s,t}}>{B_{i,t}}{P_{a,t}}}}$, $f_1\left({{\bf{c}}_{a,t}}\right)$ is a convex function, then we have a concave constraint in $(\ref{eqn_21_a})$.
\end{proof}

Based on the above analysis, the problem in (\ref{eqn_23}) is still non-convex due to the non-convex constraints in (\ref{eqn_22_b}), (\ref{eqn_21_a}) and (\ref{eqn_21_c}).
To make the problem in (\ref{eqn_23}) more tractable, the Taylor expansion is employed to approximate the convex functions with the linear ones. Then, we obtain the following lemma.
\begin{lemma}
For any  given ${{\bf{v}}_{a,t}^r}$ and ${{\bf{c}}_{a,t}^r}$, we have
\begin{flalign}
\label{eqn_14_a}&{\left\| {\bf{v}}_{a,t}^r \right\|_2^2} +2{{\bf{v}}_{a,t}^r}^T({{\bf{v}}_{a,t}} - {\bf{v}}_{a,t}^r)\geq v_{\min}^2,\\
\label{eqn_14_b}&({B_{s,t}}{P_{s,t}})^{\frac{2}{\varsigma}}f_{a,i,t}\geq ({B_{i,t}}{P_{a,t}})^{\frac{2}{\varsigma}}{\left\| {{{\bf{c}}_{a,t}} - {{\bf{c}}_{s,t}}} \right\|_2^2},\\
\label{eqn_14_c}&I_0^{\frac{2}{\varsigma}}f_{a,j,t}\ge({B_{j,t}}{P_{a,t}})^{\frac{2}{\varsigma}}
\end{flalign}
with
\begin{equation}
f_{a,i,t}={\left\|{{\bf{c}}_{a,t}^{r}-{\bf{c}}_{i,t}}\right\|_2^2+2{({{\bf{c}}_{a,t}^{r}-{\bf{c}}_{i,t}})}^T({{\bf{c}}_{a,t}} - {\bf{c}}_{a,t}^{r})}.
\end{equation}
\end{lemma}
\begin{proof}
See Appendix B.
\end{proof}

\begin{table}[!t]
\renewcommand{\arraystretch}{1.2}
\label{table_1}
  \centering
  \caption{Successive Convex Optimization of Trajectory and Transmit Power.}
  \begin{tabular}{p{0.001\columnwidth} p{0.001\columnwidth} p{0.001\columnwidth} p{0.001\columnwidth} p{0.83\columnwidth}}
    \hline
      \hline
  \multicolumn{5}{l}{\textbf{Initialization:}} \\
  &\multicolumn{4}{p{0.9\columnwidth}}{${{\bf{c}}_{a,t}^0}$, ${{\bf{v}}_{a,t}^0}$, ${\varepsilon  = 1.0 \times {10^{ - 3}}}$, $L_0=50$, $Q^{0}=0$,}\\
  \multicolumn{5}{l}{\textbf{FOR} $l=1$ \textbf{TO} $l=L_0$} \\
      &\multicolumn{4}{p{0.9\columnwidth}}{\hangafter 1 \hangindent 10pt 1)~Solve the problem in (\ref{eqn_7}) for given ${{\bf{c}}_{a,t}^{l-1}}$ and ${{\bf{v}}_{a,t}^{l-1}}$, then
      denote the optimal solution as ${P_{a,t}^{l}}$, ${{\bf{c}}_{a,t}^{l}}$, ${{\bf{v}}_{a,t}^{l}}$, ${{\bf{a}}_{a,t}^{l}}$, $Q^l$,}\\
      &\multicolumn{4}{l}{2)\ \ If ${{{\left| {{Q^l} - {Q^{l - 1}}} \right|} \mathord{\left/ {\vphantom {{\left| {{Q^l} - {Q^{l- 1}}} \right|} {{Q^{l}}}}} \right. \kern-\nulldelimiterspace} {{Q^{l}}}} < \varepsilon}$, stop.} \\
  \multicolumn{5}{l}{\textbf{END}} \\
  \hline
\end{tabular}
\end{table}

According to Lemma 1, we can iteratively solve the problem by using the successive convex optimization. The details are given in Table I. In the $l$-th iteration, by using ${{\bf{v}}_{a,t}^{l-1}}$ and ${{\bf{c}}_{a,t}^{l-1}}$ obtained in the $(l-1)$-th iteration, the optimization problem can be formulated as
\begin{flalign}\label{eqn_7}
\mathop {\max }\limits_{{P_{a,t}^l},{{\bf{c}}_{a,t}^l},{{\bf{v}}_{a,t}^l},{{\bf{a}}_{a,t}^l},Q^l}\ \ \ &Q^l \\
\rm{subject\ to}\ \ \ & (\ref{eqn_22_k}), (\ref{eqn_22_c}), (\ref{eqn_22_d}), (\ref{eqn_22_h}), (\ref{eqn_22_e}), (\ref{eqn_19}), (\ref{eqn_20}), (\ref{eqn_21_b}), (\ref{eqn_14_a}), (\ref{eqn_14_b}), (\ref{eqn_14_c})\nonumber
\end{flalign}
In constraints, the superscript $l$ is used for ${P_{a,t}},{{\bf{c}}_{a,t}},{{\bf{v}}_{a,t}}$, ${{\bf{a}}_{a,t}}$, and $Q$, respectively. Besides, ${{\bf{v}}_{a,t}^r}$ and ${{\bf{c}}_{a,t}^r}$ are replaced with ${{\bf{v}}_{a,t}^{l-1}}$ and ${{\bf{c}}_{a,t}^{l-1}}$, respectively.

\begin{figure}[!t]
  \centering
  \includegraphics[width=4 in]{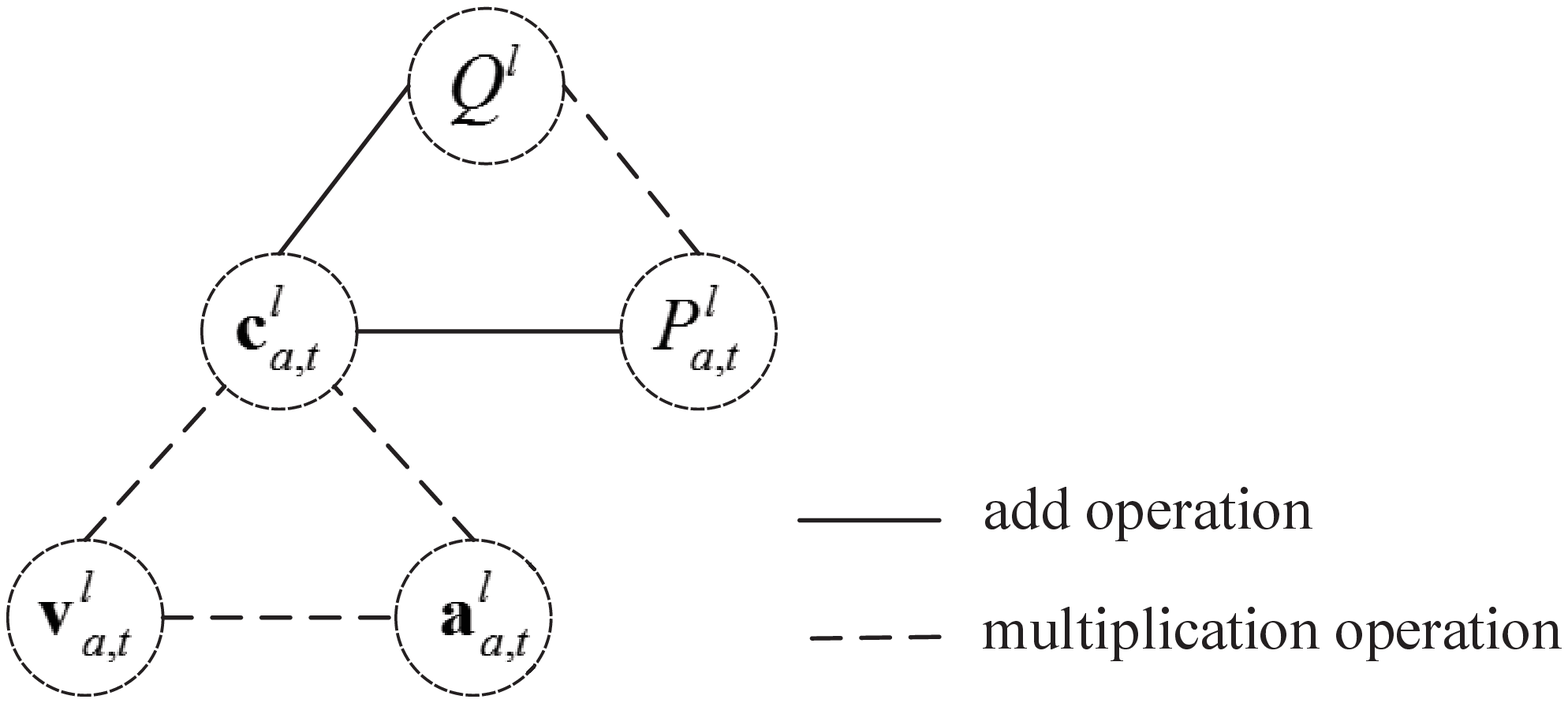}%
  \caption{Coupling relationships between the variables of the problem in (\ref{eqn_7}).}\label{fig_2}
\end{figure}

In (\ref{eqn_7}), the variables $Q^l$, ${P^l_{a,t}}$ and ${{\bf{c}}^l_{a,t}}$ are closely related to each other because of multiplication operations, as shown in Fig. \ref{fig_2}. Consequently, ${{\bf{c}}^l_{a,t}}$ cannot be obtained together with $Q^l$ and ${P^l_{a,t}}$. Geometric programming can be employed to transform the multiplication operation into add one, so that ${P^l_{a,t}}$ and ${{\bf{c}}^l_{a,t}}$ can be solved together. But it provides a tight bound. Therefore, we decouple the problem in (\ref{eqn_7}) into two subproblems, and solve it iteratively, as shown in Table II. First, with given ${{\bf{c}}^l_{a,t}}$, we optimize ${P^l_{a,t}}$. Then, with the obtained ${P^l_{a,t}}$, we optimize ${{\bf{c}}^l_{a,t}}$. In addition, due to the linear relationship, ${{\bf{c}}^l_{a,t}}$, ${{\bf{v}}^l_{a,t}}$ and ${{\bf{a}}^l_{a,t}}$ are solved together in this paper. Two subproblems are described as follow.

\subsubsection{Optimization of transmit power}
By using ${{\bf{c}}_{a,t}^{l-1}}$ obtained in the $(l-1)$-th iteration, we set ${{\bf{c}}_{a,t}^{l}={\bf{c}}_{a,t}^{l-1}}$, and optimize the transmit power ${P_{a,t}^l}$ by solving the following problem
\begin{flalign}\label{eqn_6}
\mathop {\max }\limits_{{P_{a,t}^l},Q^l}\ \ \ & Q^l \\
\nonumber \rm{subject\ to}\ \ \ &  (\ref{eqn_22_h}), (\ref{eqn_22_e}), (\ref{eqn_21_b}), (\ref{eqn_14_b}), (\ref{eqn_14_c}).
\end{flalign}
The problem in (\ref{eqn_6}) is a LP, which can be solved with CVX \cite{IEEEhowto:kopka21}.

\begin{table}[!t]
\renewcommand{\arraystretch}{1.2}
\label{table_4}
  \centering
  \caption{Successive Convex Optimization and Decoupling of Trajectory and Transmit Power.}
  \begin{tabular}{p{0.001\columnwidth} p{0.001\columnwidth} p{0.001\columnwidth} p{0.001\columnwidth} p{0.83\columnwidth}}
    \hline
      \hline
  \multicolumn{5}{l}{\textbf{Initialization:}} \\
  &\multicolumn{4}{p{0.9\columnwidth}}{${{\bf{c}}_{a,t}^0}$, ${{\bf{v}}_{a,t}^0}$, ${\varepsilon  = 1.0 \times {10^{ - 3}}}$, $L_0=50$, $Q^{0}=0$,}\\
  \multicolumn{5}{l}{\textbf{FOR} $l=1$ \textbf{TO} $l=L_0$} \\
      &\multicolumn{4}{p{0.9\columnwidth}}{\hangafter 1 \hangindent 10pt 1)~Solve the problem in (\ref{eqn_6}) for given ${{\bf{c}}_{a,t}^{l}={\bf{c}}_{a,t}^{l-1}}$, then denote the optimal solution as ${P_{a,t}^{l}}$,}\\
      & \multicolumn{4}{p{0.9\columnwidth}}{2)~Solve the problem in (\ref{eqn_2}) with given ${{\bf{c}}_{a,t}^{l- 1}}$, ${{\bf{v}}_{a,t}^{l- 1}}$, and ${P_{a,t}^{l}}$, and denote the optimal solutions as ${{\bf{c}}_{a,t}^{l}}$,${{\bf{v}}_{a,t}^{l}}$,${{\bf{a}}_{a,t}^{l}}$, $Q^l$,}\\
      &\multicolumn{4}{l}{3)\ \ If ${{{\left| {{Q^l} - {Q^{l - 1}}} \right|} \mathord{\left/ {\vphantom {{\left| {{Q^l} - {Q^{l- 1}}} \right|} {{Q^{l}}}}} \right. \kern-\nulldelimiterspace} {{Q^{l}}}} < \varepsilon}$, stop.} \\
  \multicolumn{5}{l}{\textbf{END}} \\
  \hline
\end{tabular}
\end{table}

\subsubsection{Optimization of three-dimensional coordinates, velocities and accelerations}
By using the obtained ${P^l_{a,t}}$, ${{\bf{c}}^{l-1}_{a,t}}$ and ${{\bf{v}}^{l-1}_{a,t}}$, the problem in (\ref{eqn_7}) can be rewritten as
\begin{flalign}\label{eqn_2}
\mathop {\max }\limits_{{\bf{c}}^l_{a,t},{\bf{v}}^l_{a,t},{\bf{a}}^l_{a,t},Q^l}\ \ \ & Q^l \\
\nonumber \rm{subject\ to}\ \ \ & (\ref{eqn_22_k}), (\ref{eqn_22_c}), (\ref{eqn_22_d}), (\ref{eqn_19}), (\ref{eqn_20}), (\ref{eqn_21_b}), (\ref{eqn_14_a}), (\ref{eqn_14_b}), (\ref{eqn_14_c})
\end{flalign}
Then, we can iteratively solve the problem in (\ref{eqn_7}) by employing successive convex optimization.

\begin{table}[!t]
\renewcommand{\arraystretch}{1.2}
\label{table_3}
  \centering
  \caption{Bisection Method for Solving the Problem (\ref{eqn_2}).}
  \begin{tabular}{p{0.001\columnwidth} p{0.001\columnwidth} p{0.001\columnwidth} p{0.001\columnwidth} p{0.83\columnwidth}}
    \hline
      \hline
        \multicolumn{5}{l}{\textbf{Initialization:}} \\
  &\multicolumn{4}{p{0.9\columnwidth}}{1)~${\varepsilon  = 1.0 \times {10^{ - 3}}}$, $M_0=50$,}\\
      & \multicolumn{4}{p{0.9\columnwidth}}{2)~Set ${{U^0} ={P^l_{a,t}}{B_{i,t}}z_{\min }^{ - \varsigma}}$,}\\
      \multicolumn{4}{l}{\textbf{FOR} $m=1$ \textbf{TO} $m=M_0$} \\
         & \multicolumn{3}{p{0.9\columnwidth}}{3)\ ${Q^m = {{\left( {{U^{m - 1}} + {L^{m - 1}}} \right)} \mathord{\left/ {\vphantom {{\left( {{U^{m - 1}} + {L^{m - 1}}} \right)} 2}} \right.
 \kern-\nulldelimiterspace} 2}}$,} \\
        & \multicolumn{2}{p{0.9\columnwidth}}{\hangafter 1 \hangindent 10pt 4)\ Solve the convex problem in (\ref{eqn_28}) with given ${{\bf{c}}_{a,t}^{l- 1}}$, ${{\bf{v}}_{a,t}^{l- 1}}$, ${P_{a,t}^{l}}$ and $Q^m$, and denote the optimal solutions as ${{\bf{c}}_{a,t}^{m}}$, ${{\bf{v}}_{a,t}^{m}}$, ${{\bf{a}}_{a,t}^{m}}$,} \\
         & \multicolumn{3}{p{0.9\columnwidth}}{\hangafter 1 \hangindent 10pt 5)\ If the problem is solved, ${{U^m} = {U^{m - 1}}}$, ${{L^m} = Q^m}$; otherwise ${{U^m} = Q^m}$, ${{L^m} = {L^{m - 1}}}$,}\\
         & \multicolumn{3}{p{0.9\columnwidth}}{6)\ If ${{\left| {{U^m} - {L^m}} \right|} \mathord{\left/{\vphantom {{\left| {{U^m} - {L^m}} \right|} {{L^m}}}} \right. \kern-\nulldelimiterspace} {{L^m}}} < \varepsilon $, stop,}\\
      \multicolumn{4}{l}{\textbf{END}} \\
\multicolumn{4}{l}{7)~~${Q^l=Q^m}$,} \\
      \multicolumn{4}{l}{8)~~${{\bf{c}}_{a,t}^{l}} ={{\bf{c}}_{a,t}^{m}}$, ${{\bf{v}}_{a,t}^{l}} ={{\bf{v}}_{a,t}^{m}}$, ${{\bf{a}}_{a,t}^{l}} ={{\bf{a}}_{a,t}^{m}}$.} \\
  \hline
\end{tabular}
\end{table}

Similarly, to solve the problem in (\ref{eqn_2}), the bisection method is utilized to decouple $Q^l$ and ${{\bf{c}}^l_{a,t}}$. We decompose the problem in (\ref{eqn_2}) into a series of convex problems by setting ${Q^l}$, and solve it iteratively. The details are shown in Table III. In the $m$-th iteration, let ${U^{m - 1}}$ and ${L^{m - 1}}$ respectively denote the upper bound and lower bound of ${Q^l}$. For ${Q^m = {{\left( {{U^{m - 1}} + {L^{m - 1}}} \right)} \mathord{\left/ {\vphantom {{\left( {{U^{m - 1}} + {L^{m - 1}}} \right)} 2}} \right.
 \kern-\nulldelimiterspace} 2}}$, with given ${{\bf{c}}_{a,t}^{l-1}}$, ${{\bf{v}}_{a,t}^{l-1}}$ and ${P^{l}_{a,t}}$ obtained by solving the problem in (\ref{eqn_6}), the convex problem can be formulated as
\begin{flalign}\label{eqn_28}
\textrm{find}\ \ \ &  {{\bf{c}}_{a,t}^{m}},{{\bf{v}}_{a,t}^{m}},{{\bf{a}}_{a,t}^{m}}\\
\nonumber \rm{subject\ to}\ \ \ & (\ref{eqn_22_k}), (\ref{eqn_22_c}), (\ref{eqn_22_d}), (\ref{eqn_19}), (\ref{eqn_20}), (\ref{eqn_21_b}), (\ref{eqn_14_a}), (\ref{eqn_14_b}), (\ref{eqn_14_c})
\end{flalign}
where ${{P_{a,t}},{{\bf{c}}_{a,t}},{{\bf{v}}_{a,t}},{{\bf{a}}_{a,t}},Q}$ are replaced with ${{P_{a,t}^l},{{\bf{c}}_{a,t}^m},{{\bf{v}}_{a,t}^m},{{\bf{a}}_{a,t}^m},Q^m}$, respectively. Besides, ${{\bf{v}}_{a,t}^r}$ and ${{\bf{c}}_{a,t}^r}$ are replaced with ${{\bf{v}}_{a,t}^{l-1}}$ and ${{\bf{c}}_{a,t}^{l-1}}$, respectively.
When the maximum $Q^m$ is found, with which the convex problem (\ref{eqn_28}) is solved, we achieve the related vectors ${{\bf{c}}_{a,t}^{m}},{{\bf{v}}_{a,t}^{m}},{{\bf{a}}_{a,t}^{m}}$.
The shortest distance between the UAV and the mobile user is $z_{\min}$. Given ${P^l_{a,t}}$, we set the upper bound of ${Q^1}$ to be
\begin{equation}
{{U^0} ={P^l_{a,t}}{B_{i,t}}z_{\min }^{ - \varsigma}}.
\end{equation}
The lower bound of ${Q^1}$ is set to be 0.

\section{Simulation Results and Discussion}
In this section, simulation is performed to validate the performance of our proposed algorithm. The TBS connected to the UAV is located at ${(0,~0,~100)}$ m. The UAV provides the communication services for the mobile user when the user travels from the position ${(5.0\times 10^4,~0,~10 )}$ m to ${(6.8\times 10^4,~0,~10)}$ m along x axis. We uniformly sample $T=10$ points from the positions of the user for simple analysis. The UAV flies according to the optimized trajectory. The users served by satellites and interfered by the UAV appear randomly. The interference from the UAV seriously affects the nearest users served by satellites. So without loss of generality, we set $M_t=1$. The antenna gains of the TBS and the UAV are set to be $12$ dBi and $8$ dBi. The antenna gains of the users served by the UAV and satellites are set to be $8$ dBi and $30$ dBi. The system is operated at the 5GHz carrier frequency.
The path loss is set to be
\begin{flalign}\label{eqn_60}
{L_{i,j,t}}\left( {{\rm{dB}}} \right) ={116.7} + 15 \log 10\left( {\frac{{{d_{i,j,t}}}}{{{2600}}}} \right) + {X_{i,j,t}}
\end{flalign}
where the standard deviation of $X_{i,j,t}$ is 0.1. The main parameters are given in Table IV.
For each experiment, we randomly generate the small-scale fading for 1000 rounds to achieve ergodic achievable rates according to the parameters given in Table IV.

\begin{table}[!t]
\renewcommand{\arraystretch}{1.2}
\label{table_2}
  \centering
  \caption{Simulation Parameters.}
\begin{tabular}{|l|l|l|l|}
  \hline
  \textbf{Symbol}       & \textbf{Value}                          & \textbf{Symbol}         & \textbf{Value}\\ \hline
  $z_{\min}$  & 2.6~\textrm{km}                & $v_{\min}$   & 10~\textrm{m/s}   \\ \hline
  $z_{\max}$  & 5.0~\textrm{km}                & $v_{\max}$   & 60~\textrm{m/s}   \\ \hline
  $v_i$       & 30~\textrm{m/s}                & $P_{s,t}$    & 40~\textrm{dBm}   \\ \hline
  $\sigma ^2$ & -107~\textrm{dBm}              & $a_{\max}$   & $10~\textrm{m/s}^2$  \\
  \hline
\end{tabular}
\end{table}

\subsection{Performance Comparison among Different Algorithms}
In this part, we compare our proposed algorithm with those in \cite{IEEEhowto:kopka3} and \cite{IEEEhowto:kopka14}. In these works, the full CSI was used for the whole trajectory optimization. Let ${\bf{c}}_{i,t}=[x_{i,t},~y_{i,t},~z_{i,t}]^T$ be the positions of the user served by the UAV, and $v_i$ be the user's velocity. For comparison, we adopt a basic trajectory which is denoted as ${\bf{c}}_{i,t}=[x_{i,t},~y_{i,t},~z_{\min}]^T$. The transmit power is set to satisfy the constraints on tolerable interference, backhaul, maximum transmit power and the total communication energy of the UAV. Besides, the positions of the users served by satellites are set as ${{\bf{c}}_{j,t}=[x_{i,t},~y_{i,t}+(-1)^t\times8000,~z_{i,t}]^T}$ with $v_i=30 m/s$. The initial trajectory of the UAV is set to be ${\bf{c}}_{a,t}=[x_{i,t}/2,~y_{i,t},~z_{\min}]^T$.

\begin{figure*}[!t]
  \centering
  \includegraphics[width=5 in]{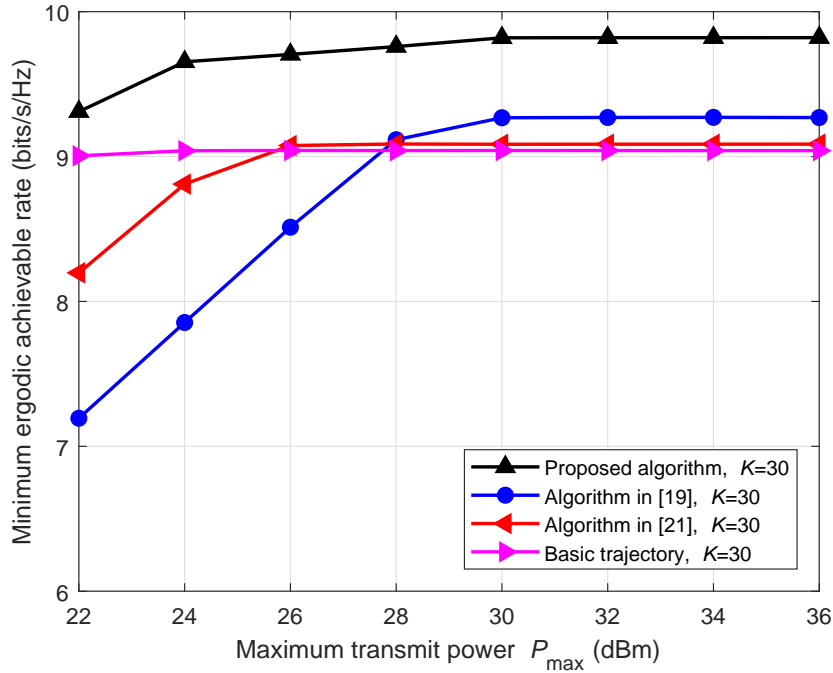}
  \caption{Minimum ergodic achievable rate of different algorithms with Rician factor $K=30$, the interference temperature limitation $I_0={-40}$ dBm and the total communication energy $E_0=500$~J.}\label{fig_8}
\end{figure*}
\begin{figure*}[!t]
  \centering
  \includegraphics[width=5 in]{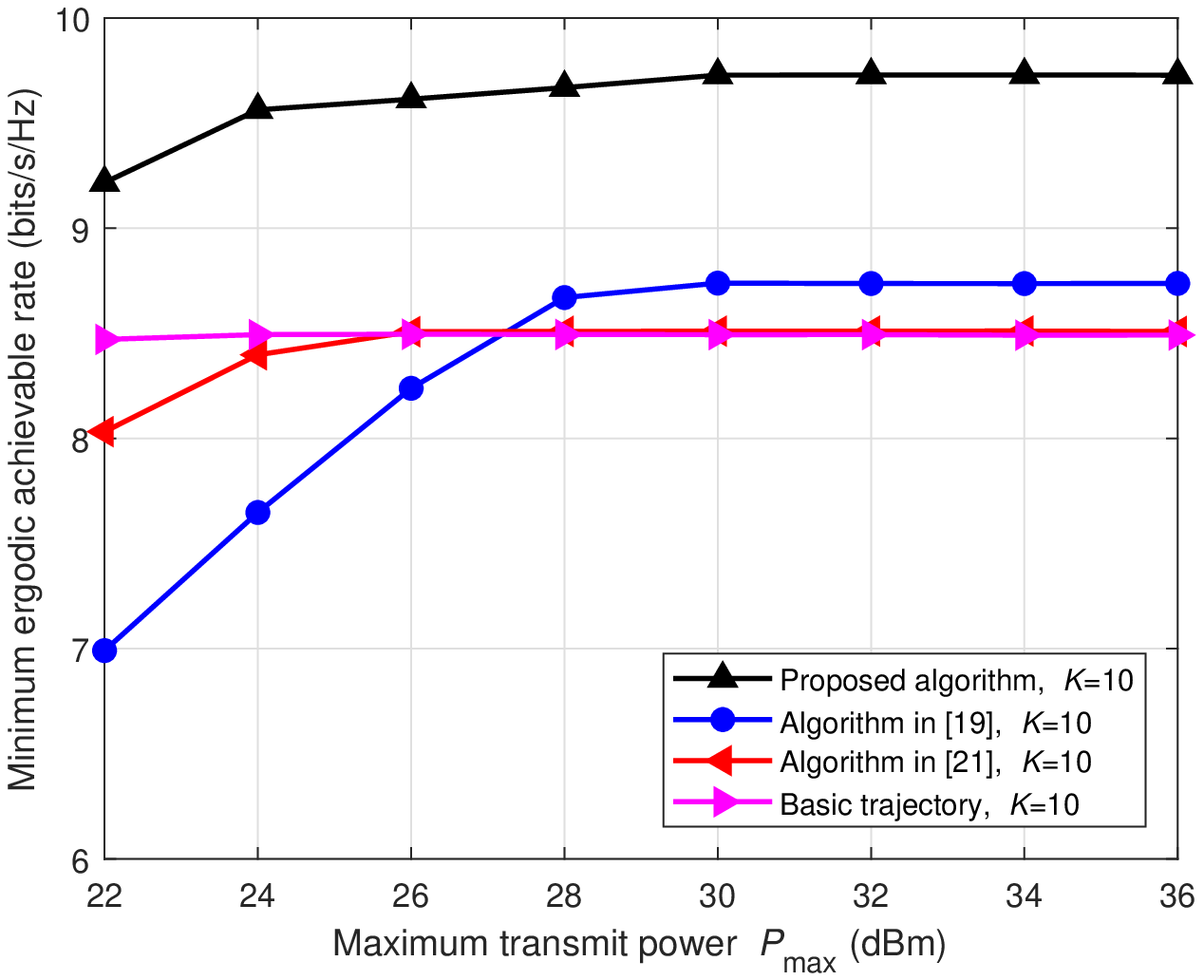}
  \caption{Minimum ergodic achievable rate of different algorithms with Rician factor $K=10$, the interference temperature limitation $I_0={-40}$ dBm and the total communication energy $E_0=500$~J.}\label{fig_12}
\end{figure*}

Because of the difficulty of obtaining the small-scale CSI, the full CSI can not be accurately obtained in practice. In our proposed algorithm, the whole trajectory and the transmit power of the UAV are optimized with the large-scale CSI only. To validate the performance of our proposed algorithm, the minimum ergodic achievable rate of different algorithms is compared. The simulation results are shown in Fig. \ref{fig_8}, where $E_0$ is 500 J. We set that the interference temperature limitation $I_0$ is $-40$~dBm, and vary maximum transmit power $P_{\max}$ in the range $[22,~36]$ dBm. Because $I_0$ is large, the interference can be ignored. The transmit power is bounded by the maximum transmit power, backhaul and total communication energy. When $P_{\max}\leq 30$ dBm, the performance is mainly determined by backhaul and maximum transmit power. The existing algorithms ignore the constraint of maximum transmit power. We decrease their transmit power to satisfy this constraint. One sees that the performance can be improved with the optimization problem subject to the constraint of maximum transmit power. When $P_{\max}\geq 30$ dBm, the total transmit power during $T$ is larger than the total communication energy, and the performance is mainly determined by backhaul and total communication energy. The algorithm in \cite{IEEEhowto:kopka3} investigated the optimization problem with full CSI subject to constraints of backhaul and total communication energy. Our proposed algorithm achieves better performance than that in \cite{IEEEhowto:kopka3}. To further validate the performance of our proposed algorithm using the large-scale CSI, we vary Rician factor $K$. The simulation results are shown in Fig. \ref{fig_12}. One sees that by reducing $K$, our proposed algorithm obtains much better performance than the existing ones. One sees that the performance can be improved with the large-scale CSI.

\begin{figure}[!t]
  \centering
   \includegraphics[width=5 in]{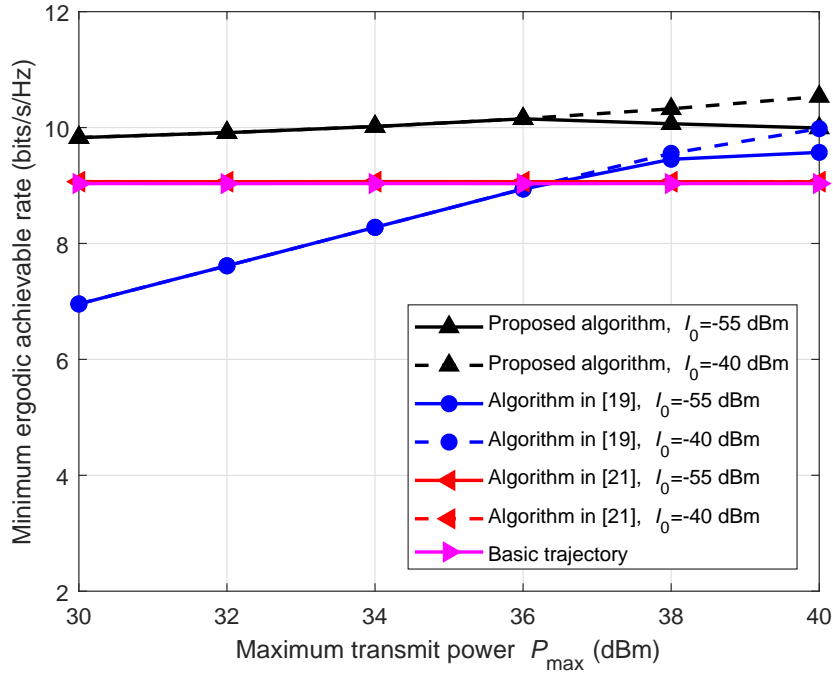}%
  \caption{Minimum ergodic achievable rate of different algorithms with the interference temperature limitation $I_0={-55}$~dBm or ${-40}$ dBm and the total communication energy $E_0=3\times10^4$~J.}\label{fig_7}
\end{figure}

To illustrate the performance gain achieved by using interference constraint, the comparison of minimum ergodic achievable rate is shown in Fig. \ref{fig_7}, where $K=31.3$. We set $E_0=3\times10^4$ J. Because $E_0$ is large, the transmit power is limited by interference, maximum transmit power and backhaul. We set that the interference temperature limitation $I_0$ is $-55$~dBm and $-40$~dBm, and vary maximum transmit power $P_{\max}$ in the range $[30,~40]$ dBm. When $I_0=-40$~dBm, the interference can be ignored. The algorithms in \cite{IEEEhowto:kopka3} and \cite{IEEEhowto:kopka14} neglect the constraints of interference and maximum transmit power. We reduce their transmit power to satisfy those constraints. By varying $I_0$ and $P_{\max}$, the minimum ergodic achievable rate is increased when $P_{\max}\geq 36$ dBm. One sees that the transmit power is determined by interference constraint when $P_{\max}\geq 36$ dBm and $I_0=-55$~dBm. The performance of our proposed algorithm is best of all when $P_{\max}\geq 36$ dBm and $I_0=-55$~dBm. Thus, our proposed algorithm can improve minimum ergodic achievable rate by a joint optimization of the whole trajectory and transmit power with interference constraints.

\subsection{Discussion on the Impact of Key Parameters}
\begin{figure}[!t]
  \centering
   \includegraphics[width=5 in]{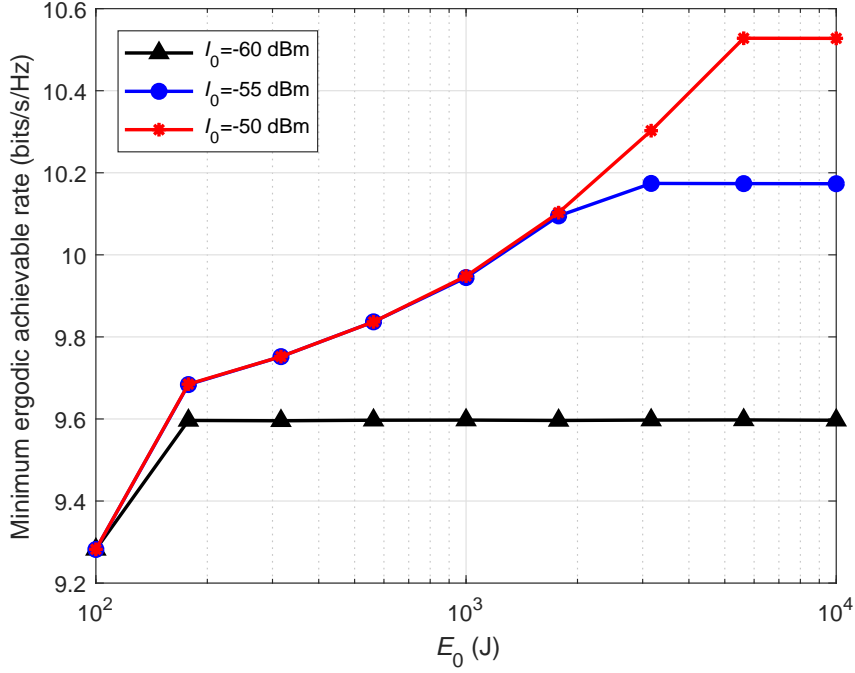}%
  \caption{Minimum ergodic achievable rate with different interference temperature limitation $I_0$.}\label{fig_9}
\end{figure}
\begin{figure}[!t]
  \centering
   \includegraphics[width=5 in]{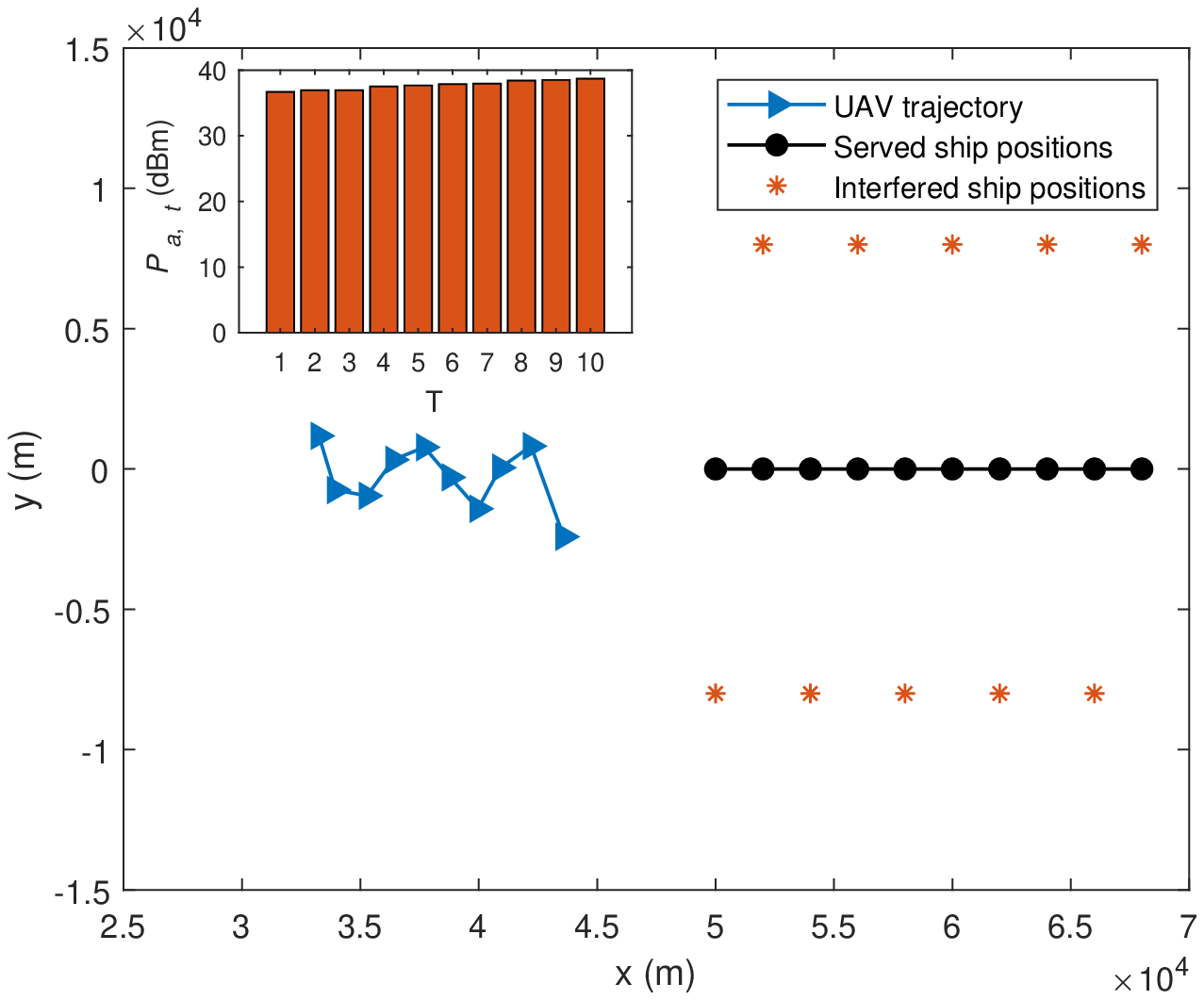}%
  \caption{Optimized trajectory in the x-y plane.}\label{fig_10}
\end{figure}

In this part, we analyze the impact of total energy and the interference on minimum ergodic achievable rate. Set maximum transmit power $P_{\max}=40$ dBm. The simulation results are shown in Fig. \ref{fig_9}, where the total energy $E_0$ is in the range $[100,~10000]$~J. The interference temperature limitation $I_0$ is set to be $-60$~dBm, $-55$~dBm and $-50$~dBm, respectively. The initial trajectory of UAV is $[x_{i,t},~y_{i,t},~z_{\min}]^T$,  $[3x_{i,t}/4,~y_{i,t},~z_{\min}]^T$ and $[x_{i,t}/2,~y_{i,t},~z_{\min}]^T$, respectively.  When $I_0$ and $E_0$ are increased, better performance can be obtained. When the energy constraint is tight, the performance is determined by $E_0$. By increasing $E_0$, when the interference constraint is tight, the performance is determined by $I_0$.

An optimized trajectory in the x-y plane is shown in Fig. \ref{fig_10}, where $P_{\max}=40$ dBm, $I_0=-55$ dBm, $E_0=4000$ J, and $K=31.3$. Because of constraints of wireless backhaul, the optimized trajectory is between TBS and the mobile user. Besides, the optimized trajectory was bent to satisfy interference constraints. The obtained transmit power of UAV satisfies the constraints of maximum transmit power and total communication energy.

\begin{figure}[!t]
  \centering
  \includegraphics[width=5 in]{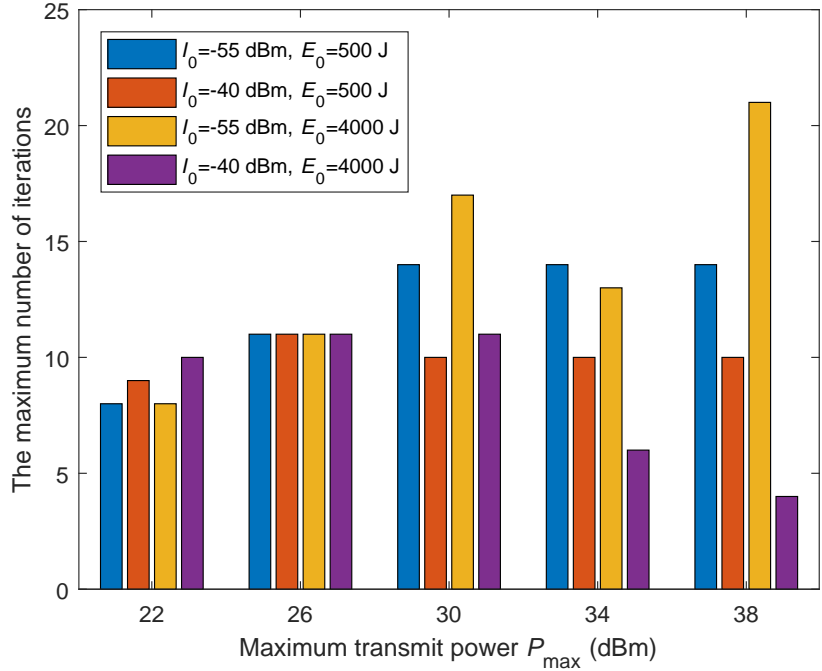}%
  \caption{Maximum number of iterations.}\label{fig_11}
\end{figure}

\subsection{Convergence Performance of the Proposed Algorithm}
The convergence is analyzed in this part. The experiment is implemented 100 rounds by generating different scenes. In each scene, the users served by satellites and interfered by the UAV appear randomly. The distance between the users served by satellites and the one served by UAV is 8000 m. The maximum numbers of iterations are shown in Fig. \ref{fig_11}, where maximum transmit power $P_{\max}$ in the range $[22,~38]$ dBm, the interference temperature limitation $I_0$ is $-55$~dBm and $-40$~dBm, and the total energy $E_0$ is 500~J and 4000~J. One sees that, the maximum number of iterations is smaller than 21. Thus, the algorithm converges within $21$ iterations.

\section{Conclusions}
In this paper, UAVs have been used for on-demand satellite-terrestrial maritime communications. The coordination with existing satellites/terrestrial systems has been investigated to realize spectrum sharing and efficient backhaul. This paper has adopted a typical composite channel model consisting of both large-scale and small-scale fading, under which UAVs have been deployed for accompanying coverage.
The UAV's whole trajectory and transmit power during the fight have been jointly optimized, subject to constraints on UAV kinematics, tolerable interference, backhaul, and the total communication energy of the UAV. Different from previous studies, we have assumed that only the large-scale CSI is available, as the positions of mobile ships can be obtained via the maritime AIS and be used as the prior information. Then, we have solved the non-convex problem by problem decomposition, successive convex optimization and bisection searching tools. Simulation results have shown that the UAV fits well with existing satellite and terrestrial systems. Besides, the performance gain can be achieved via joint optimization of the UAV's trajectory and transmit power with only the large-scale CSI.

\appendices
\section{Proof of Theorem 1}
Since ${g_{a,i,t} \in \mathcal{CN}(0,1)}$, the average SNR can be achieved and denoted as
\begin{equation}
{\bf{E}}\left\{ {{P_{a,t}}{G_a}{G_i}|{h_{a,i,t}}{|^2}{\sigma ^{ - 2}}} \right\}= {P_{a,t}}{G_a}{G_i}L_{a,i,t}^{ - 1}{\sigma ^{ - 2}}.
\end{equation}
Let ${a_{a,i,t}} = {P_{a,t}}{G_a}{G_i}L_{a,i,t}^{ - 1}{\sigma ^{ - 2}}$. We analyze the relationship between ${R_{a,i,t}}$ and ${a_{a,i,t}}$ via the first-order and second-order derivatives.
By using the known positions of the transmitter and the receiver, the path loss ${L_{a,i,t}}$ in (\ref{eqn_10}) can be obtained. Let
\begin{equation}
{b_{a,i,t}} = {\left| {\sqrt {\frac{K}{{1 + K}}}  + \sqrt {\frac{1}{{1 + K}}} {g_{a,i,t}}} \right|^2}.
\end{equation}
Since ${g_{a,i,t} \in \mathcal{CN}(0,1)}$, the variable ${b_{a,i,t}}$ follows a non-central chi-square probability density function with two degrees of freedom as
\begin{equation}
{f_{{b_{a,i,t}}}}\left(\gamma \right) = \left({1 + K} \right){e^{- K}}{e^{- \left( {1 + K} \right)\gamma }}{I_0}\left( {2\sqrt{K\left({1 + K}\right)\gamma}} \right)
\end{equation}
where $\gamma\geq 0$ and ${I_0({\cdot})}$ is the zeroth-order modified Bessel function of the first kind \cite{IEEEhowto:kopka23}. Then, ${R_{a,i,t}}$ can be expressed as
\begin{flalign}\label{eqn_9}
{R_{a,i,t}} = {\rm{lo}}{{\rm{g}}_2}e\int_0^\infty  {{\rm{ln}}\left( {1 + {a_{a,i,t}}\gamma } \right){f_{{b_{a,i,t}}}}\left( \gamma  \right)d} \gamma.
\end{flalign}
The first-order derivative with respect to ${a_{a,i,t}}$ is
\begin{flalign}\label{eqn_17}
{{\dot R}_{a,i,t}} ={\rm{lo}}{{\rm{g}}_2}e\int_0^\infty  {\frac{\gamma }{{1 + {a_{a,i,t}}\gamma }}{f_{{b_{a,i,t}}}}\left( \gamma  \right)d} \gamma.
\end{flalign}
The second-order derivative with respect to ${a_{a,i,t}}$ is
\begin{flalign}\label{eqn_18}
{\ddot R_{a,i,t}} = {\rm{lo}}{{\rm{g}}_2}e\int_0^\infty  {\frac{{ - {\gamma ^2}}}{{{{\left( {1 + {a_{a,i,t}}\gamma } \right)}^2}}}{f_{{b_{a,i,t}}}}\left( \gamma  \right)d} \gamma.
\end{flalign}
Because ${a_{a,i,t}}\geq 0$ and ${f_{{b_{a,i,t}}}}\left( \gamma  \right)>0$, ${{\dot R}_{a,i,t}}>0$ and ${{\ddot R}_{a,i,t}}<0$.
So, ${R_{a,i,t}}$ is an increasing function of ${a_{a,i,t}}$ and strictly concave.

Thus, the theorem is proved.

\section{Proof of Lemma 1}
According to that any convex function is globally lower-bounded by its first-order Taylor expansion at any point \cite{IEEEhowto:kopka18}, with the given ${{\bf{v}}_{a,t}^r}$ and ${{\bf{c}}_{a,t}^r}$, we have the following inequalities
\begin{flalign}\label{eqn_24}
&{\left\| {{\bf{v}}_{a,t}} \right\|_2^2} \ge  {\left\| {\bf{v}}_{a,t}^r \right\|_2^2} +
2{{\bf{v}}_{a,t}^r}^T({{\bf{v}}_{a,t}} - {\bf{v}}_{a,t}^r),\\
&{\left\| {{\bf{c}}_{a,t}-{\bf{c}}_{i,t}} \right\|_2^2} \ge  {\left\|{{\bf{c}}_{a,t}^r-{\bf{c}}_{i,t}}\right\|_2^2} +
2{({{\bf{c}}_{a,t}^r-{\bf{c}}_{i,t}})}^T({{\bf{c}}_{a,t}} - {\bf{c}}_{a,t}^r).
\end{flalign}
Then, combining the constraints in (\ref{eqn_22_b}), (\ref{eqn_21_a}) and (\ref{eqn_21_c}), the lemma is proved.

\bibliographystyle{IEEEtran}

\end{spacing}
\end{document}